\documentclass{llncs}
\usepackage{graphicx}

\graphicspath{{./graphics/}}

\begin{document}

\title{Lower Bounds for the Size of Nondeterministic Circuits}

\author{Hiroki Morizumi}
\institute{Interdisciplinary Graduate School of Science and Engineering, Shimane University, Shimane 690-8504, Japan\\
\email{morizumi@cis.shimane-u.ac.jp}}

\maketitle

\begin{abstract}
Nondeterministic circuits are a nondeterministic computation model in
circuit complexity theory.
In this paper, we prove a $3(n-1)$ lower bound for the size of nondeterministic
$U_2$-circuits computing the parity function.
It is known that the minimum size of (deterministic) $U_2$-circuits computing
the parity function exactly equals $3(n-1)$.
Thus, our result means that nondeterministic computation is useless to
compute the parity function by $U_2$-circuits and cannot reduce the size
from $3(n-1)$.
To the best of our knowledge, this is the first nontrivial lower bound
for the size of nondeterministic circuits (including formulas, constant
depth circuits, and so on) with unlimited nondeterminism for an
explicit Boolean function.
We also discuss an approach to proving lower bounds for the size of
deterministic circuits via lower bounds for the size of nondeterministic
restricted circuits.
\end{abstract}

\section{Introduction}

Proving lower bounds for the size of Boolean circuits is a central topic
in circuit complexity theory.
The gate elimination method is one of well-known proof techniques to prove
lower bounds for the size of Boolean circuits, and has been used to prove
many linear lower bounds including the best known lower bounds for the size
of Boolean circuits over the basis $B_2$~\cite{B84}\cite{DK11} and
the basis $U_2$~\cite{LR01}\cite{IM02}.

In this paper, we show that the gate elimination method works well also for
nondeterministic circuits.
For deterministic circuits, it is known that the minimum size of
$U_2$-circuits computing the parity function exactly equals
$3(n-1)$~\cite{S74}.
The proof of the lower bound is based on the gate elimination method
and has been known as a typical example that the method is effective.
In this paper, we prove a $3(n-1)$ tight lower bound for the size of
nondeterministic $U_2$-circuits computing the parity function, which means that
nondeterministic computation is useless to compute the parity function
by $U_2$-circuits and cannot reduce the size from $3(n-1)$.

To the best of our knowledge, our result is the first nontrivial lower bound
for the size of nondeterministic circuits (including formulas, constant
depth circuits, and so on) with unlimited nondeterminism for an
explicit Boolean function.
In this paper, we show that, for $U_2$-circuits, a known proof technique
(i.e., the gate elimination method) for deterministic circuits is applicable
to the nondeterministic case.
This implies the possibility that proving lower bounds for the size of
nondeterministic circuits may not be so difficult in contrast with the intuition
and known proof techniques might be applicable to the nondeterministic case
also for other circuits.
One of motivations to prove lower bounds for the size of nondeterministic
circuits is from some relations between the size of deterministic
circuits and nondeterministic circuits.
We also discuss an approach to proving lower bounds for the size of
deterministic circuits via lower bounds for the size of nondeterministic
restricted circuits.

\section{Preliminaries}

\subsection{Definitions}

{\em Circuits} are formally defined as directed acyclic graphs.
The nodes of in-degree 0 are called {\em inputs}, and each one of them
is labeled by a variable or by a constant 0 or 1.
The other nodes are called {\em gates}, and each one of them
is labeled by a Boolean function.
The {\em fan-in} of a node is the in-degree of the node, and
the {\em fan-out} of a node is the out-degree of the node.
There is a single specific node called {\em output}.

We denote by $B_2$ the set of all Boolean functions
$f:\{0,1\}^2 \rightarrow \{0,1\}$.
By $U_2$ we denote $B_2 - \{\oplus, \equiv\}$, i.e., $U_2$ contains
all Boolean functions over two variables except for the XOR function and
its complement.
A Boolean function in $U_2$ can be represented as the following form:
$$f(x,y) = ((x \oplus a) \wedge (y \oplus b)) \oplus c,$$
where $a, b, c \in \{0,1\}$.
A {\em $U_2$-circuit} is a circuit in which each gate has fan-in 2
and is labeled by a Boolean function in $U_2$.
A {\em $B_2$-circuit} is similarly defined.

A {\em nondeterministic circuit} is a circuit with {\em actual inputs}
$(x_1, \ldots, x_n) \in \{0,1\}^n$ and some further inputs
$(y_1, \ldots, y_m) \in \{0,1\}^m$ called {\em guess inputs}.
A nondeterministic circuit computes a Boolean function $f$ as follows:
For $x \in \{0,1\}^n$, $f(x)=1$ iff there exists a setting of the guess inputs
$\{y_1, \ldots, y_m\}$ which makes the circuit output 1.
In this paper, we call a circuit without guess inputs
a {\em deterministic circuit} to distinguish it
from a nondeterministic circuit.

The {\em size} of a circuit is the number of gates in the circuit.
The {\em depth} of a circuit is the length of the longest path from an input
to the output in the circuit.
We denote by $size^{\rm dc}(f)$ the size of the smallest deterministic
$U_2$-circuit computing a function $f$, and denote by $size^{\rm ndc}(f)$
the size of the smallest nondeterministic $U_2$-circuit computing
a function $f$.

While we mainly consider $U_2$-circuits in this paper,
we also consider other circuits in Section~\ref{sec:disc}.
A {\em formula} is a circuit whose fan-out is restricted to 1.
The parity function of $n$ inputs $x_1, \ldots, x_n$, denoted by Parity$_n$,
is 1 iff $\sum x_i \equiv 1 \pmod{2}$.

\subsection{The gate elimination method} \label{subsec:method}

The proof of our main result is based on the gate elimination method,
and based on the proof of the deterministic case.
In this subsection, we have a quick look at them.

Consider a gate $g$ which is labeled by a Boolean function in $U_2$.
Recall that any Boolean function in $U_2$ can be represented as
the following form:
$$f(x,y) = ((x \oplus a) \wedge (y \oplus b)) \oplus c,$$
where $a, b, c \in \{0,1\}$.
If we fix one of two inputs of $g$ so that $x = a$ or $y = b$,
then the output of $g$ becomes a constant $c$.
In such case, we call that $g$ is {\em blocked}.

\begin{theorem}[Schnorr~\cite{S74}] \label{thrm:par}
$$size^{\rm dc}({\rm Parity}_n) = 3(n-1).$$
\end{theorem}

\begin{proof}
Assume that $n \geq 2$.
Let $C$ be an optimal deterministic $U_2$-circuit computing Parity$_n$.
Let $g_1$ be a top gate in $C$, i.e.,  whose two inputs are connected from
two inputs $x_i$ and $x_j$, $1 \leq i, j \leq n$.
Then, $x_i$ must be connected to another gate $g_2$, since, if $x_i$ is
connected to only $g_1$, then we can block $g_1$ by an assignment
of a constant to $x_j$ and the output of $C$ becomes independent from $x_i$,
which contradicts that $C$ computes Parity$_n$.
By a similar reason, $g_1$ is not the output of $C$.
Let $g_3$ be a gate which is connected from $g_1$.
See Figure~\ref{fig:sch}.

\begin{figure}[t]
  \begin{center}
    \hspace{-2cm}
    \includegraphics[scale=0.6]{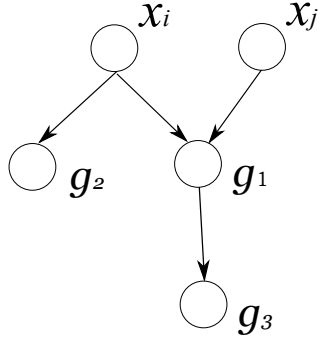}
    \caption{Proof of Theorem~\ref{thrm:par}}
    \label{fig:sch}
  \end{center}
\end{figure}

We prove that we can eliminate at least 3 gates from $C$ by an assignment
to $x_i$.
We assign a constant 0 or 1 to $x_i$ such that $g_1$ is blocked.
Then, we can eliminate $g_1$, $g_2$ and $g_3$.
If $g_2$ and $g_3$ are the same gate, then the output of $g_2$ ($= g_3$) 
becomes a constant, which means that $g_2$ ($= g_3$) is not the output
of $C$ and we can eliminate another gate which is connected from
$g_2$ ($= g_3$).
Thus, we can eliminate at least 3 gates and the circuit come to compute
Parity$_{n-1}$ or $\neg$Parity$_{n-1}$.
For deterministic circuits, it is obvious that 
$size^{\rm dc}({\rm Parity}_{n-1}) = size^{\rm dc}(\neg{\rm Parity}_{n-1})$.
Therefore,
\begin{eqnarray*}
size^{\rm dc}({\rm Parity}_n) & \geq   & size^{\rm dc}({\rm Parity}_{n-1}) + 3 \\
                            & \vdots & \\
                            & \geq   & 3(n-1).
\end{eqnarray*}

$x \oplus y$ can be computed with 3 gates by the following form:
$$(x \wedge \neg y) \vee (\neg x \wedge y).$$
Therefore, $size^{\rm dc}({\rm Parity}_n) \leq 3(n-1)$. \qed
\end{proof}

\section{Proof of the Main Result}

In this section, we prove the main theorem.
For deterministic circuits, there must be a top gate whose two inputs
are connected from two (actual) inputs $x_i$ and $x_j$, $1 \leq i, j \leq n$.
However, for nondeterministic circuits, there may be no such gate,
since there are not only actual inputs but also guess inputs
in nondeterministic circuits.
We need to defeat the difficulty.
See Section~\ref{subsec:method} for the definition of ``block''.

\begin{theorem} \label{thrm:main}
$$size^{\rm ndc}({\rm Parity}_n) = 3(n-1).$$
\end{theorem}

\begin{proof}
By theorem~\ref{thrm:par},
$$size^{\rm ndc}({\rm Parity}_n) \leq size^{\rm dc}({\rm Parity}_n) = 3(n-1).$$

Assume that $n \geq 2$.
Let $C$ be an optimal nondeterministic $U_2$-circuit computing Parity$_n$.
We prove that we can eliminate at least 3 gates from $C$ by an assignment
of a constant 0 or 1 to an actual input.

\bigskip

\noindent Case 1. There is an actual input $x_i$, $1 \leq i \leq n$, which
is connected to at least two gates.

\smallskip

Let $g_1$ and $g_2$ be gates which are connected from $x_i$.
Since we can block $g_1$ by an assignment of a constant to $x_i$,
$g_1$ is not the output of $C$ and there is a gate $g_3$ which is connected
from $g_1$.
See Figure~\ref{fig:case1}.

We prove that we can eliminate at least 3 gates from $C$ by an assignment
to $x_i$.
We assign a constant 0 or 1 to $x_i$ such that $g_1$ is blocked.
Then, we can eliminate $g_1$, $g_2$ and $g_3$.
If $g_2$ and $g_3$ are the same gate, then the output of $g_2$ ($= g_3$) 
becomes a constant, which means that $g_2$ ($= g_3$) is not the output
of $C$ and we can eliminate another gate which is connected from
$g_2$ ($= g_3$).
Thus, we can eliminate at least 3 gates.

\begin{figure}[t]
  \begin{center}
    \includegraphics[scale=0.6]{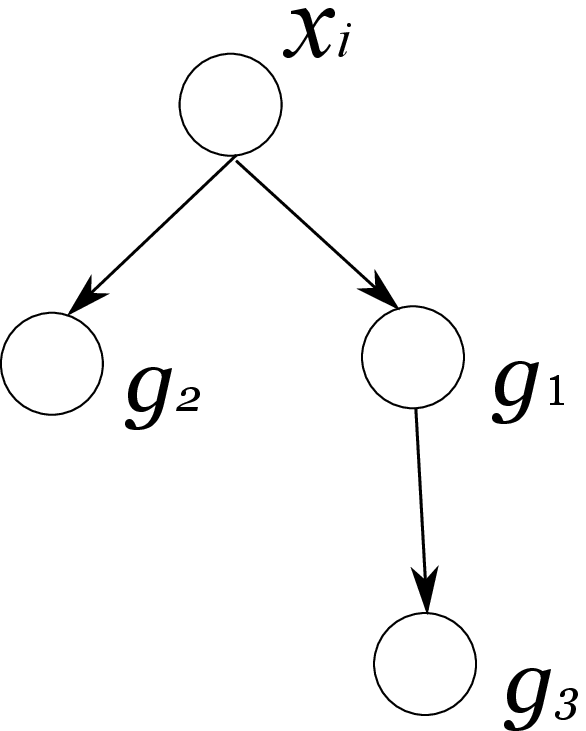}
    \caption{Case 1}
    \label{fig:case1}
  \end{center}
\end{figure}

\bigskip

\noindent Case 2. Every actual input is connected to at most one gates.

\smallskip

Let $g_1$ be a gate in $C$ such that one of two inputs is connected from
an actual input $x_i$ and the other is connected from a node $v$ whose
output is dependent on only guess inputs and independent from actual inputs.
($v$ may be a gate and may be a guess input.)
Consider that an assignment to actual inputs and guess inputs is given.
Then, if the value of the output of $v$ blocks $g_1$ by the assignment,
then the output of $C$ must be 0, since, if the output of $C$ is 1,
then the value of the Boolean function which is computed by $C$ becomes
independent from $x_i$, which contradicts that $C$ computes Parity$_n$.
(Note that the output of $C$ can be 0. The difference is from the definition
 of nondeterministic circuits.)
We use the fact above and reconstruct $C$ as follows.

Let $c$ be a constant 0 or 1 such that if the output of $v$ is $c$, then
$g_1$ is blocked.
We fix the input of $g_1$ from $v$ to $\neg c$ and eliminate $g_1$.
We prepare a new output gate $g_2$ and connect the two inputs of $g_2$
from the old output gate and $v$.
$g_2$ is labeled by a Boolean function in $U_2$ so that the output of $g_2$
is 1 iff the input from the old output gate is 1 and the input from $v$ is
$\neg c$.
Let $C'$ be the reconstructed circuit.
See Figure~\ref{fig:case2}.

In the reconstruction, we eliminated one gate ($g_1$) and added one gate
($g_2$). Thus, the size of $C'$ equals the size of $C$.
In $C'$, if the output of $v$ is $c$, then the output of $C'$ becomes 0
by $g_2$.
If the output of $v$ is $\neg c$, then the output of $C'$ equals the output
of the old output gate and $g_1$ has been correctly eliminated since we
fixed the input of $g_1$ from $v$ to $\neg c$ in the reconstruction.
Thus, $C$ and $C'$ compute a same Boolean function.
We repeat such reconstruction until the reconstructed circuit satisfies
the condition of Case 1.
The repetition must be ended, since one repetition increases continuous
gates whose one input is dependent on only guess inputs (i.e., $g_2$)
at the output.

\begin{figure}[t]
  \begin{center}
    \includegraphics[scale=0.5]{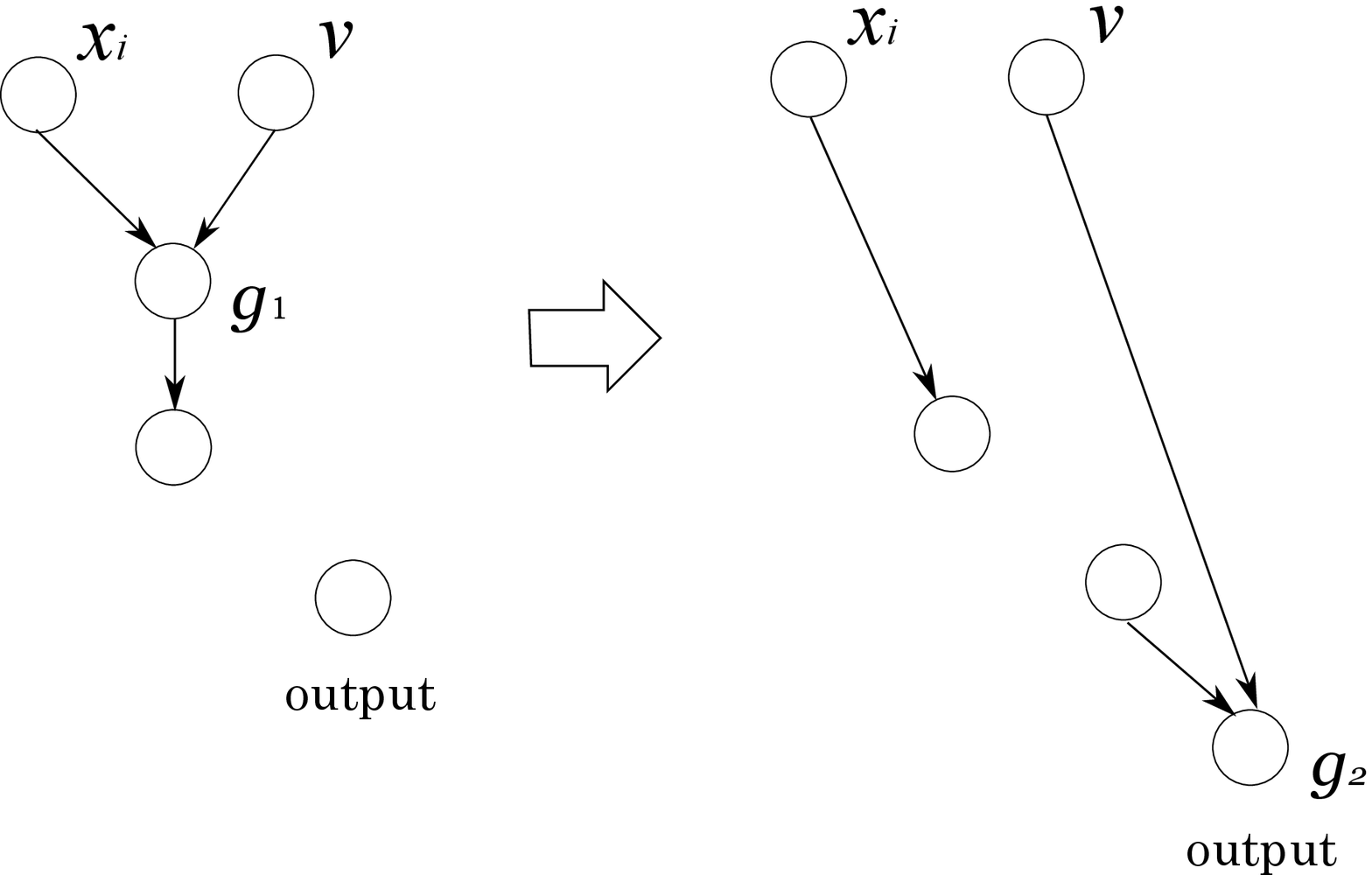}
    \caption{Case 2}
    \label{fig:case2}
  \end{center}
\end{figure}

\bigskip

Thus, we can eliminate at least 3 gates for all cases and the circuit
come to compute Parity$_{n-1}$ or $\neg$Parity$_{n-1}$.

\begin{lemma}
$$size^{\rm ndc}({\rm Parity}_{n-1}) = size^{\rm ndc}(\neg{\rm Parity}_{n-1}).$$
\end{lemma}

\begin{proof}
Let $C'$ be a nondeterministic $U_2$-circuit computing Parity$_{n-1}$.
For nondeterministic circuits, notice that negating the output gate in $C'$
does not give a circuit computing $\neg$Parity$_{n-1}$.
We negate an arbitrary actual input $x_i$, $1 \leq i \leq n$, in $C'$
and obtain a nondeterministic $U_2$-circuit which computes
$\neg$Parity$_{n-1}$ and has the same size to $C'$, which can be done
by relabeling each Boolean function of all gates which are connected
from $x_i$. \qed
\end{proof}

Therefore,

\begin{eqnarray*}
size^{\rm ndc}({\rm Parity}_n) & \geq   & size^{\rm ndc}({\rm Parity}_{n-1}) + 3 \\
                            & \vdots & \\
                            & \geq   & 3(n-1).
\end{eqnarray*}

Thus, the theorem holds. \qed

\end{proof}

\section{Discussions} \label{sec:disc}

In this paper, we proved a $3(n-1)$ lower bound for the size of nondeterministic
$U_2$-circuits computing the parity function.
To the best of our knowledge, this is the first nontrivial lower bound
for the size of nondeterministic circuits with unlimited nondeterminism
for an explicit Boolean function.
In this section, as one of motivations to prove lower bounds for the size
of nondeterministic circuits, we discuss an approach to proving lower bounds
for the size of deterministic circuits via lower bounds for the size of
nondeterministic restricted circuits.

It is known that the Tseitin transformation~\cite{T68}
converts an arbitrary Boolean circuit to a CNF formula.
We restate the Tseitin transformation as the form of the following theorem.

\begin{theorem} \label{thrm:conv1}
Any $B_2$-circuit of $n$ inputs and size $s$ can be converted to
a nondeterministic 3CNF formula of $n$ actual inputs,
$s$ guess inputs, and size $O(s)$.
\end{theorem}

\begin{proof}
We prepare one guess input for the output of each gate in the $B_2$-circuit,
and use the Tseitin transformation. \qed
\end{proof}

Thus, if we hope to prove a nonlinear lower bound for the size of deterministic
general circuits, (which is a major open problem in circuit complexity theory,)
then it is enough to prove a nonlinear lower bound for the size of
nondeterministic circuits in the theorem.
The nondeterministic circuit in Theorem~\ref{thrm:conv1} is a constant depth
circuit with depth two and some restrictions.
In this paper, we saw that, for $U_2$-circuits, a known proof technique
(i.e., the gate elimination method) for deterministic circuits is applicable
to the nondeterministic case.
It remains future work whether many known ideas or techniques for constant depth
circuits are applicable to nondeterministic circuits.

The basic idea of the proof of Theorem~\ref{thrm:conv1} and
the Tseitin transformation can be widely applied.
We show another example in which guess inputs are prepared for a part of
gates in the circuit.

\begin{theorem}
Any $B_2$-circuit of $n$ inputs, size $O(n)$ and depth $O(\log n)$
can be converted to a nondeterministic formula of $n$ actual inputs,
$O(n / \log\log n)$ guess inputs, and size $O(n^{1+\epsilon})$,
where $\epsilon > 0$ is an arbitrary small constant.
\end{theorem}

\begin{proof}
Let $C$ be such a $B_2$-circuit.
It is known that we can find $O(n / \log\log n)$ edges in $C$ whose
removal yields a circuit of depth at most $\epsilon \log n$
(\cite{V76}, Section~14.4.3 of \cite{AB09}).
We prepare $O(n / \log\log n)$ guess inputs for the edges, and
one guess input for the output of $C$.
Consider that an assignment to actual inputs and guess inputs is given.
It can be checked whether the value of each guess input corresponds
to correct computation by $O(n / \log\log n)$ nondeterministic formulas
of size $n^{\epsilon}$, since the depth is at most $\epsilon \log n$.
We construct a nondeterministic formula so that it outputs 1 iff all guess
inputs are correct and the guess input which corresponds to the output of
$C$ is 1. \qed
\end{proof}

For the case that the number of guess inputs is limited,
there is a known lower bound for the size of nondeterministic
formulas~\cite{K98}.

\bibliographystyle{splncs03}
\bibliography{circuit}

\end{document}